\documentclass[a4paper,runningheads]{llncs}

\usepackage[utf8]{inputenc} 

\title{On Induced Online Ramsey Number of Paths,\\ Cycles, and Trees}
\author{V{\'a}clav Bla{\v z}ej\inst{1} \and Pavel Dvo{\v r}{\'a}k\inst{2}
\thanks{P. Dvo{\v r}{\' a}k was supported by the project GAUK 1514217.}
\and Tom{\'a}{\v s} Valla\inst{1}}
\institute{
Faculty of Information Technology, Czech Technical University in Prague,\\ Prague, Czech Republic
\and
Charles University, Faculty of Mathematics and Physics, Prague, Czech Republic
}

\usepackage{graphicx} 
\usepackage{subfig} 
\usepackage{amsmath} 
\usepackage{amssymb} 
\usepackage{tikz} 
\usepackage{xspace} 
\usepackage{thmtools} 
\usepackage{mathtools}
\usepackage{algorithm}
\usepackage{todonotes}

\usetikzlibrary{calc, shapes, backgrounds, patterns, arrows, decorations,
                decorations.pathreplacing, decorations.pathmorphing, fit}

\tikzset{
    >=stealth',
    longpath/.style={decorate, decoration={snake}},
    main/.style={draw=black,circle,inner sep=2pt,solid},
    selected red/.style={draw,line width=5pt,-,red!20},
    selected blue/.style={draw,line width=5pt,-,blue!20},
    hide/.style = {fill = none, draw = none, shape=rectangle},
    red/.style = {draw = red, solid, semithick},
    blue/.style = {draw = blue, dashed, semithick},
    tobe/.style = {draw = black, dotted, semithick},
    both/.style = {draw,dashed,blue,dash pattern=on 3pt off 3pt,
        postaction={dashed,red,dash pattern=on 7pt off 11pt, dash phase=-1pt} },
    gray/.style = {draw=gray, very thin},
    black/.style = {draw = black}
}

\def\R{\mathbb{R}}

\newcommand{\WLOG}{without loss of generality\xspace}

\begin{document}

\maketitle

\def\rf{\overline{r}} 
\def\ri{r_{ind}} 
\def\rsif{\rf_{ind}} 
\def\ro{\widetilde{r}} 
\def\rsio{\ro_{ind}}
\def \rrf [#1,#2]{\rf_{#1}(#2)} 
\def \rro [#1,#2]{\ro_{#1}(#2)} 

\def\spread[#1,#2]{#1_1, #1_2, \dots, #1_{#2} } 
\def\spreadu[#1,#2]{#1^1, #1^2, \dots, #1^{#2} } 
\newcommand{\comp}[1]{\overline{#1}} 

\hyphenation{de-ve-loped}
\hyphenation{com-bi-na-to-ri-al}
\hyphenation{Fried-gut}




\begin{abstract}
An online Ramsey game is a game between Builder and
Painter, alternating in turns.
They are given a graph $H$ and a graph $G$ of an infinite set of independent vertices.
In each round Builder draws an edge and Painter colors it either red or blue.
Builder wins if after some finite round there is a monochromatic
copy of the graph $H$, otherwise Painter wins.
The online Ramsey number $\ro(H)$ is the minimum number of rounds such that Builder can
force a monochromatic copy of $H$ in $G$.
This is an analogy to the size-Ramsey number $\rf(H)$
defined as the minimum number such that there exists graph $G$ with $\rf(H)$ edges
where for any edge two-coloring $G$ contains a monochromatic copy of $H$.

In this paper, we introduce the concept of induced online Ramsey numbers:
the induced online Ramsey number $\rsio(H)$ is the minimum number of rounds Builder
can force an induced monochromatic copy of $H$ in $G$.
We prove asymptotically tight bounds on the induced online Ramsey numbers of paths,
cycles and two families of trees.
Moreover, we provide a result analogous to Conlon [On-line Ramsey Numbers, SIAM J. Discr. Math. 2009],
showing that there is an infinite family of trees $T_1,T_2,\dots$, $|T_i|<|T_{i+1}|$ for $i\ge1$, such that
\[
    \lim_{i\to\infty} \frac{\ro(T_i)}{\rf(T_i)} = 0.
\]
\end{abstract}

\section{Introduction}

For a graph $H$, the Ramsey number $r(H)$ is the smallest integer $n$ such that
in any two-coloring of edges of the complete graph $K_n$, there is a monochromatic copy of $H$.
The size-Ramsey number $\rf(H)$, introduced by Erd\H os, Faudree, Rousseau, and Schelp~\cite{Erdos1978TheSRN},
is the smallest integer $m$ such that there exists a graph $G$ with $m$ edges such that
for any two-coloring of the edges of $G$ one will always find a monochromatic copy of $H$.

There are many interesting variants of the usual Ramsey function.
One important concept is the \emph{induced Ramsey number} $\ri(H)$,
which is the smallest integer $n$ for which there is a graph $G$ on $n$ vertices
such that every edge two-coloring of $G$ contains an induced monochromatic copy of $H$.
Erd\H os~\cite{Erdos75problemsand} conjectured the existence of a constant $c$
such that every graph $H$ with $n$ vertices satisfies $\ri(H)\le 2^{cn}$,
which would be best possible.
In 2012, Conlon, Fox and Sudakov~\cite{Conlon2012} proved that there is a constant $c$
such that every graph $H$ with $n$ vertices satisfies $\ri(H)\le 2^{cn\log n}$.
The proof uses a construction of explicit pseudorandom graphs, as opposed to
random graph construction techniques used by previous attempts.
For more on the topic see the excellent review by Conlon, Fox, and Sudakov~\cite{Conlon2015RecentDI}.

The induced size-Ramsey number $\rsif(H)$ is an analog of the size-Ramsey number:
we define $\rsif(H)$ as the smallest integer $m$ such that there exists a graph $G$ with $m$ edges
such that for any two-coloring of the edges of $G$ there is always a monochromatic copy of $H$.
In 1983, Beck~\cite{Beck1983size}, using probabilistic methods, proved the surprising fact that
$\ro(P_n)\le cn$, where $P_n$ is a path of length $n$ and $c$ is an absolute constant.
An even more surprising result came by Haxell, Kohayakawa, and \L uczak~\cite{Haxell1995InducedCycles},
who studied the induced size-Ramsey number of cycles showing that $\rsif(C_n)=O(n)$.
However, the proof uses random graph techniques and regularity lemma
and does not provide any reasonably small multiplicative constant.

\smallskip

In this paper, we study the online variant of size Ramsey number which was introduced independently
by Beck~\cite{Beck1993AchievementsGA} and Kurek and Ruci\'nski~\cite{Kurek2005TwoVariants}.
The best way to define it is in term of a game between two players, Builder and Painter.
An infinite set of vertices is given, in each round Builder draws a new edge and immediately
it is colored by Painter in either red or blue.
The goal of Builder is to force Painter to obtain a monochromatic copy of a fixed graph $H$ (called \emph{target graph}).
The minimum number of edges which Builder must draw in order to obtain such monochromatic
copy of $H$, assuming optimal strategy of Painter, is known as the online Ramsey number $\ro(H)$.
The graph $G$, which is being built by Builder, is called \emph{background graph}.
The online Ramsey number is guaranteed to exist because Builder can simply
create a big complete graph $K_{r(H)}$, which by Ramsey theorem trivially contains a monochromatic copy of $H$.

The winning condition for Builder is to obtain a copy of the target graph $H$.
However, there are more different notions of ``being a copy''.
This leads us to the following two definitions.

\begin{itemize}
\item The \emph{online Ramsey number} $\ro(H)$ is the minimum number of rounds of the Builder-Painter game Builder has
a strategy to obtain a monochromatic subgraph $H$.
\item The \emph{(strongly) induced online Ramsey number} $\rsio(H)$ is the minimum number of rounds of the Builder-Painter game
such that Builder has a strategy to obtain a monochromatic induced subgraph $H$ in $G$.
\end{itemize}
If there is no strategy of Builder to obtain the copy of $H$, we define the respective number as $\infty$.

Note that for any graph $H$ we have $\ro(H) \leq \rsio(H)$.
Also, note that the induced online Ramsey numbers provide lower bounds on the induced size-Ramsey numbers.

In 2008 Grytczuk, Kierstead and Pra{\l}at~\cite{Grytczuk2008OnlineRN} studied the online Ramsey number of paths,
obtaining $\ro(P_n)\le 4n-3$, where $P_n$ is a path with $n$ edges, providing an interesting
counterpart to the result of Beck~\cite{Beck1983size}.
Also, the result by Haxell, Kohayakawa, and {\L}uczak.~\cite{Haxell1995InducedCycles} on induced size-Ramsey number
of cycles naturally bounds the online version as well, but with no reasonable multiplicative constant.

In this paper, we study the induced online Ramsey number of paths, cycles, and trees.
The summary of the results for paths and cycles is as follows.
\begin{theorem}
Let $P_n$ denote the path of length $n$ and let $C_n$ denote a cycle with $n$ vertices.
Then
\begin{itemize}
\item $\rsio(P_n) \le 28n-27$,
\item $\rsio(C_n) \le 367n-27$ for even $n$,
\item $\rsio(C_n) \le 735n-27$ for odd $n$.
\end{itemize}
\end{theorem}
A \emph{spider} $\sigma_{k,\ell}$ is a union of $k$ paths of length $\ell$ sharing
exactly one common endpoint. We further show that $\rsio(\sigma_{k,\ell}) = \Theta(k^2\ell)$
and $\ro(\sigma_{k,\ell}) = \Theta(k^2\ell)$.

Although we know that $\ro(H)\le\rf(H)$,
it is a challenging task to identify classes of graphs for which there
is an asymptotic gap between both numbers.
For complete graphs, Chv\'atal observed (see \cite{Erdos1978TheSRN}) that $\rf(K_t)=\binom{r(K_t)}{2}$.
The basic question, attributed to R\"odl (see \cite{Kurek2005TwoVariants}),
is to show $\lim_{t\to\infty} \ro(K_t)/\rf(K_t)$, or put differently, to show that $\ro(K_t) = o(\binom{r(K_t)}{2})$.
This conjecture remains open,
but in 2009 Conlon~\cite{Conlon2009OnlineRN} showed there exists $c>1$ such that for infinitely many $t$,
$$
\ro(K_t) \le c^{-t}\binom{r(K_t)}{2}.
$$
In this paper we contribute to this topic by showing that
there is an infinite family of trees $T_1,T_2,\dots$, with $|T_i|<|T_{i+1}|$ for $i\ge1$, such that
$$
\lim_{i\to\infty} \frac{\ro(T_i)}{\rf(T_i)} = 0,
$$
thus exhibiting the desired asymptotic gap. In fact, we prove a stronger statement,
exhibiting the asymptotic gap even for the induced online Ramsey number.

\section{Induced paths}
\label{sec:inducedPaths}
In this section we present an upper bound on the induced online Ramsey number of paths.
\begin{theorem}\label{thm:indPaths}
    Let $P_n$ be a path of length $n$.
    Then $\rsio(P_n) \le 28n-27$.
\end{theorem}
{
\begin{proof} 
\def \R [#1,#2]{R^{#1}}
\def \B [#1,#2]{B^{#1}}
\def \P [#1]{p^{#1}}
\def \RR [#1]{R^{#1}}
\def \BB [#1]{B^{#1}}
\def \CONC {\cup}
First we build the set $I$ of $2(7n-7)-1$ isolated edges,
then at least $7n-7$ have the same color, we say this color is \emph{abundant} in $I$.

Let $\R[0,0]$ and $\B[0,0]$ be the initial paths of lengths $0$.
In $s$-th step we have a red induced path $\R[s,r]=(r_0,\{r_0,r_1\},r_1,\dots,r_r)$ of length $r$
and a blue induced path $\B[s,b]=(b_0,\{b_0,b_1\},b_1,\dots,b_b)$ of length $b$.
We denote the concatenation of paths $A$ and $B$ by $A\CONC B$.
The removal of vertices and incident edges is denoted by $A \setminus \{v\}$.
We define the potential of $s$-th step $\P[s]=3a+4o$ where $a$ is the length of the path in color
which is abundant in $I$ and $o$ is the length of path in the other color.
Further, we show that we are able to maintain the invariant
that there are no edges between the $\R[s,r]$ and $\B[s,b]$ and that $p^{s+1} > p^s$.

Assume \WLOG that the blue edges are abundant in $I$.
Let $g=\{x,y\}$ be an unused blue edge from the set $I$.
One step of Builder is as follows.
Builder creates an edge $e=\{r_r,b_b\}$.
If Painter colored $e$ red then Builder creates an edge $f=\{b_b,x\}$,
however if $e$ is blue then Builder creates $f=\{r_r,x\}$.

Depending on how the $e$ and $f$ edges were colored we end up with four different scenarios.
These different cases are also depicted in Fig.~\ref{fig:paths}.

\bgroup
\small
\[
    (\BB[s+1],\RR[s+1]) = 
    \begin{cases} 
        \bigl(\B[s,b] \CONC (e,r_r,f,x,g,y), \R[s,r] \setminus \{r_r,r_{r-1}\}\bigr) & \text{if $e$ and $f$ are blue} \\
        \bigl(\B[s,b] \setminus \{b_b\}, \R[s,r] \CONC (f,x)\bigr) & \text{if $e$ is blue and $f$ is red} \\
        \bigl(\B[s,b] \CONC (f,x,g,y), \R[s,r] \setminus \{r_r\}\bigr) & \text{if $e$ is red and $f$ is blue} \\
        \bigl(\B[s,b] \setminus \{b_b,b_{b-1}\}, \R[s,r] \CONC (e,b_b,f,x)\bigr) & \text{if $e$ and $f$ are red}
    \end{cases}
\]
\[
    \P[s+1] = 
    \begin{cases} 
        3\bigl(|B^s|+3\bigr)+4\bigl(|R^s|-2\bigr) = \P[s]+1 &\text{if $e$ and $f$ are blue} \\
        3\bigl(|B^s|-1\bigr)+4\bigl(|R^s|+1\bigr) = \P[s]+1 &\text{if $e$ is blue and $f$ is red} \\
        3\bigl(|B^s|+2\bigr)+4\bigl(|R^s|-1\bigr) = \P[s]+2 &\text{if $e$ is red and $f$ is blue} \\
        3\bigl(|B^s|-2\bigr)+4\bigl(|R^s|+2\bigr) = \P[s]+2 &\text{if $e$ and $f$ are red}
    \end{cases}
\]
\egroup
We obtain a pair of paths $\BB[s+1], \RR[s+1]$ such that $\P[s+1]>\P[s]$ and invariant holds.

The maximum potential for which Builder did not win yet is $\P[s] = 7n-7$.
Therefore there are no more than $7n-6$ steps to finish one monochromatic induced path of length $n$.
To create the initial set $I$ Builder creates $2(7n-7)-1$ isolated edges.
In each step, Builder creates two edges.
The total number of edges created by Builder is no more than $2(7n-6)+2(7n-7)-1 = 28n-27$.
\qed\end{proof}
}

{%
\def\u{${b_b}$}
\def\v{${r_r}$}
\def\w{$w$}
\def\ea{$e$}
\def\eb{$f$}
\def\ec{$g$}
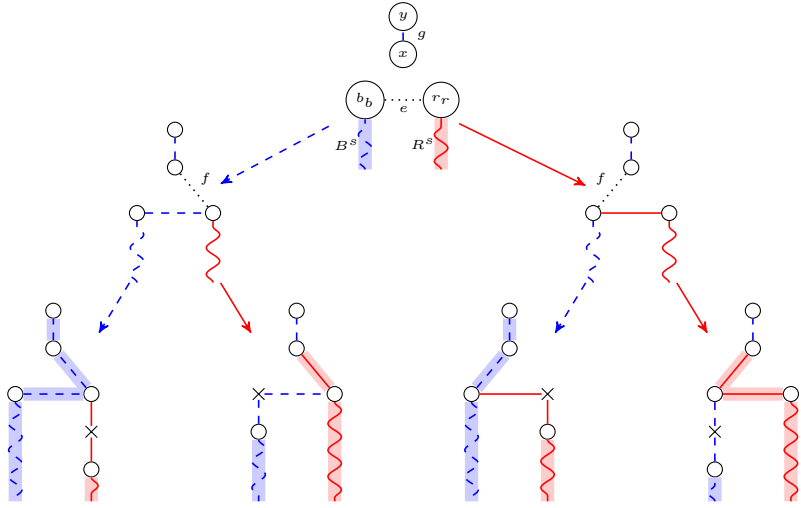
\begin{figure}[h]
    \centering
    \begin{tikzpicture}[scale=0.5]
    \tikzstyle{every node}=[font=\tiny, main]
    \def\xy{1.21}
    \def\yy{2.21}
    \def\zeroY{8}
    \def\firstShift{6}
    \def\firstY{5}
    \def\secondShift{3.2}
    \def\secondY{0.2}
    \begin{scope}[local bounding box=scope0, shift={(0,\zeroY)}]
        \node (b) at (-1,0){\u};
        \node (r) at (1,0){\v};
        \node (x) at (0,\xy){$x$};
        \node (y) at (0,\yy){$y$};
        \node[hide] (b2) at (-1,-2){};
        \node[hide] (r2) at ( 1,-2){};
        \draw[blue] (x) --node[hide,right,shift={(0.1,0)}]{\ec} (y);
        \draw[selected blue]  (b2) -- (b);
        \draw[blue, longpath] (b2) -- node [hide,left] {$B^s$} (b);
        \draw[selected red]   (r2) -- (r);
        \draw[red, longpath]  (r2) -- node [hide,left] {$R^s$} (r);
        \draw[tobe] (b) -- node [hide,below] {\ea} (r);
    \end{scope}
    \begin{scope}[local bounding box=scope1, shift={(-\firstShift,\firstY)}]
        \node (b) at (-1,0){};
        \node (r) at (1,0){};
        \node (x) at (0,\xy){};
        \node (y) at (0,\yy){};
        \node[hide] (b1) at (-1,-2){};
        \node[hide] (r1) at (1,-2){};
        \draw[blue, longpath] (b1) -- (b);
        \draw[red, longpath]  (r1) -- (r);
        \draw[blue] (b) -- (r);
        \draw[blue] (x) -- (y);
        \draw[tobe] (r) --node[hide,above right]{\eb} (x);
    \end{scope}
    \begin{scope}[local bounding box=scope2, shift={(\firstShift,\firstY)}]
        \node (b) at (-1,0){};
        \node (r) at (1,0){};
        \node (x) at (0,\xy){};
        \node (y) at (0,\yy){};
        \node[hide] (b1) at (-1,-2){};
        \node[hide] (r1) at (1,-2){};
        \draw[blue, longpath] (b1) -- (b);
        \draw[red, longpath]  (r1) -- (r);
        \draw[red] (b) -- (r);
        \draw[blue] (x) -- (y);
        \draw[tobe] (b) --node[hide,above left]{\eb} (x);
    \end{scope}
    \begin{scope}[local bounding box=scope11, shift={(-\firstShift-\secondShift,\secondY)}]
        \node (b) at (-1,0){};
        \node (r) at (1,0){};
        \node (x) at (0,\xy){};
        \node (y) at (0,\yy){};
        \node[cross out] (r1) at (1,-1){};
        \node (r2) at (1,-2){};
        \node[hide] (b3) at (-1,-3){};
        \node[hide] (r3) at ( 1,-3){};
        \draw[selected red]   (r2) -- (r3);
        \draw[red]            (r) -- (r1);
        \draw[red]            (r1) -- (r2);
        \draw[red, longpath]  (r2) -- (r3);
        \draw[selected blue]  (b3) -- (b);
        \draw[selected blue]  (b) -- (r);
        \draw[selected blue]  (r) -- (x);
        \draw[selected blue]  (x) -- (y);
        \draw[blue, longpath] (b) -- (b3);
        \draw[blue]           (b) -- (r);
        \draw[blue]           (r) -- (x);
        \draw[blue]           (x) -- (y);
    \end{scope}
    \begin{scope}[local bounding box=scope12, shift={(-\firstShift+\secondShift,\secondY)}]
        \node[cross out] (b) at (-1,0){};
        \node (r) at (1,0){};
        \node (x) at (0,\xy){};
        \node (y) at (0,\yy){};
        \node (b1) at (-1,-1){};
        \node[hide]           (b2) at (-1,-3){};
        \node[hide]           (r2) at ( 1,-3){};
        \draw[selected blue]  (b1) -- (b2);
        \draw[blue]           (b) -- (b1);
        \draw[blue, longpath] (b1) -- (b2);
        \draw[selected red]   (r) -- (r2);
        \draw[red, longpath]  (r) -- (r2);
        \draw[selected red]   (r) -- (x);
        \draw[red]            (r) -- (x);
        \draw[blue]           (b) -- (r);
        \draw[blue]           (x) -- (y);
    \end{scope}
    \begin{scope}[local bounding box=scope21, shift={(+\firstShift-\secondShift,\secondY)}]
        \node (b) at (-1,0){};
        \node[cross out] (r) at (1,0){};
        \node (x) at (0,\xy){};
        \node (y) at (0,\yy){};
        \node (r1) at (1,-1){};
        \node[hide] (b2) at (-1,-3){};
        \node[hide] (r2) at ( 1,-3){};
        \draw[selected red]   (r1) -- (r2);
        \draw[red]            (r) -- (r1);
        \draw[red, longpath]  (r1) -- (r2);
        \draw[selected blue]  (b2) -- (b);
        \draw[blue, longpath] (b) -- (b2);
        \draw[red]            (b) -- (r);
        \draw[selected blue]  (b) -- (x);
        \draw[blue]           (b) -- (x);
        \draw[selected blue]  (x) -- (y);
        \draw[blue]           (x) -- (y);
    \end{scope}
    \begin{scope}[local bounding box=scope22, shift={(+\firstShift+\secondShift,\secondY)}]
        \node (b) at (-1,0){};
        \node (r) at (1,0){};
        \node (x) at (0,\xy){};
        \node (y) at (0,\yy){};
        \node[cross out] (b1) at (-1,-1){};
        \node (b2) at (-1,-2){};
        \node[hide] (b3) at (-1,-3){};
        \node[hide] (r3) at ( 1,-3){};
        \draw[selected blue]  (b3) -- (b2);
        \draw[blue, longpath] (b3) -- (b2);
        \draw[blue          ] (b2) -- (b1);
        \draw[blue          ] (b1) -- (b);
        \draw[selected red]   (r) -- (r3);
        \draw[red, longpath]  (r) -- (r3);
        \draw[selected red]   (r) -- (b);
        \draw[selected red]   (x) -- (b);
        \draw[red]            (r) -- (b);
        \draw[red]            (x) -- (b);
        \draw[blue]           (x) -- (y);
    \end{scope}
    \begin{scope}[black,->]
        \draw[blue]  (scope0) -> (scope1);
        \draw[red]   (scope0) -> (scope2);
        \draw[blue]  (scope1) -> (scope11);
        \draw[red]   (scope1) -> (scope12);
        \draw[blue]  (scope2) -> (scope21);
        \draw[red]   (scope2) -> (scope22);
    \end{scope}
\end{tikzpicture}
    \caption{One step in creating an induced monochromatic $P_n$}
    \label{fig:paths}
\end{figure}
}

Note that the initial edges each span $2$ vertices and in each step only the first edge can lead to a new vertex.
This gives us bound on the number of vertices used in creating an induced path $P_n$
to be at most $2\bigl(2(7n-7)-1\bigr) + 7n-6 = 35n-36$.


\section{Cycles and Induced Cycles}
\label{sec:cycles}
\bgroup 

\def\asOdd{24n-20}
\def\asOddInd{735n-27}
\def\asEvenInd{367n-27}
\def\P{\rho}

In this section, we present a constructive upper bound on the online Ramsey number of cycles $\ro(C_n)$ and induced cycles $\rsio(C_n)$.

\begin{theorem}\label{thm:evencycles}
    Let $C_n$ be a cycle on $n$ vertices, where $n$ is even.
    Then, $\rsio(C_n) \le \asEvenInd$.
\end{theorem}

\begin{proof}
First, Builder obtains disjoint paths $\spread[\P,9]$ of length $4n/3-1$ and one path $\P_{10}$ of length $n-2$.
Instead of using Theorem~\ref{thm:indPaths} to create these paths separately
it is more efficient to create a $P_{13n}$ using at most $28(13n)-27$ edges
and define paths $\spread[\P,10]$ as an induced subgraph of $P_{13n}$.
Let the $P_{13n}$ be \WLOG red.
Let $\P_{i,j}$ denote the $j$-th vertex of $\P_i$.

\newcommand{\PAR}{p}
Builder will create a red $C_n$ using $\spread[\P,10]$ or three blue paths of length $n/2$ starting in $u$ and ending in either $\P_{10,1}$ or $\P_{10,n-1}$.
These three paths starting in the same vertex and two of them sharing a common endpoint will form a blue $C_n$.
Each blue path will go through a separate triple of paths from $\spread[\P,9]$ and
alternate between them with each added vertex.

Let us run the following procedure three times -- once for each $k \in \{1,2,3\}$.
Let $p=\P_{3k-2}$, $q=\P_{3k-1}$ and $r=\P_{3k}$.
Let us define cyclic order of these paths to be $p,q,r,p$ which defines a natural successor for each path.
Builder does the following three steps, which are also depicted in Fig.~\ref{fig:even_cycles_small}.
\begin{enumerate}
    \item Create edges $\{u,p_{1}\}$ and $\{u,p_{n-1}\}$. If both of these edges are red Builder wins immediately. If that is not the case then at least one edge $\{u,v_1\}$ where $v_1\in\{p_{1}, p_{n-1}\}$ is blue.
    \item Now for $i$ from $1$ to $n/2-1$ we do as follows:
\begin{itemize}
    \item Let $j := 2\lfloor{i/3}\rfloor$. Let $t \in \{p,q,r\}$ such that $v_i \in t$ and set $s$ to be the successor of $t$. 
    \item We create edges $\{v_i,s_{j+1}\}$ and $\{v_i,s_{j+n-1}\}$. If both are red Builder wins, otherwise take an edge $\{v_i,v_{i+1}\}$ where $v_{i+1} \in \{s_{j+1}, s_{j+n-1}\}$ is blue.
\end{itemize}
    \item Finish the path $(u,\spread[v,n/2-1])$ by creating edges $\{v_{n/2-1},\P_{10,1}\}$ and $\{v_{n/2-1},\P_{10,n-1}\}$. Again if both edges are red, Builder wins immediately. Otherwise, Builder creates a blue path from $u$ to $\P_{10,1}$ or to $\P_{10,n-1}$.
\end{enumerate}
\begin{figure}[h]
    \centering
    \begin{tikzpicture}[scale=0.45]
        \node[main] (u) at (-3,5) {$u$};
        \def\goal{18};
        \def\mx{20};

        \foreach \s[count=\si from 0] in {0,1,...,\mx}{
            \node[main] (p\si) at (\s,4){};
            \node[main] (q\si) at (\s,2){};
            \node[main] (r\si) at (\s,0){};
        }
        \foreach \s[count=\si from 2, count=\ss from 4] in {2,...,\goal}{
            \node[main] (x\s) at (\ss,-2.5){};
        }
        \foreach \s[count=\si from 2, count=\ss from 3] in {3,...,\goal}{
            \draw[red] (x\si) -- (x\ss);
        }
        \foreach \s[count=\si from 0, count=\ss from 1] in {1,...,\mx}{
            \draw[red] (p\si) -- (p\ss);
            \draw[red] (q\si) -- (q\ss);
            \draw[red] (r\si) -- (r\ss);
        }
        \node[main,draw=gray] (rcorner) at (\mx+1,3) {};
        \node[main,draw=gray] (lcorner) at (-1,1) {};

        \node[hide] () at (-1,4){$p$};
        \node[hide] () at (-1,2){$q$};
        \node[hide] () at (-1,0){$r$};
        \node[hide] () at (3,-2.5){$\rho_{10}$};
        \draw[dotted] (p\mx) -- (rcorner) -- (q\mx);
        \draw[dotted] (q0) -- (lcorner) -- (r0);
        \draw[blue] (u) edge [out=-20,in=140] (p0);
        \draw[dotted] (u) edge [out=-10,in=155] (p16);
        \draw[blue] (p0) -- (q0) -- (r0);
        \draw[dotted] (p0) edge[out=-20,in=160] (q16);
        \draw[blue] (r0) edge[bend right=20] (p2);
        \draw[blue] (p2) -- (q2) -- (r2);
        \draw[blue] (r2) edge[bend right=20] (p4);
        \draw[dotted] (p4) -- (q4);
        \draw[blue] (p4) edge[out=340,in=150] (q20);
        \draw[blue] (q20) edge[out=260,in=40] (x2);
        \draw[dotted] (q20) edge[out=280,in=120] (x\goal);
        \draw[dotted] (r20) -- (\mx+1,0);
        \draw[dotted] (x\goal) -- (\mx+1,-2.5);
    \end{tikzpicture}
    \caption{Creation of $\P_{n/2}$ for $n=18$.}
    \label{fig:even_cycles_small}
\end{figure}
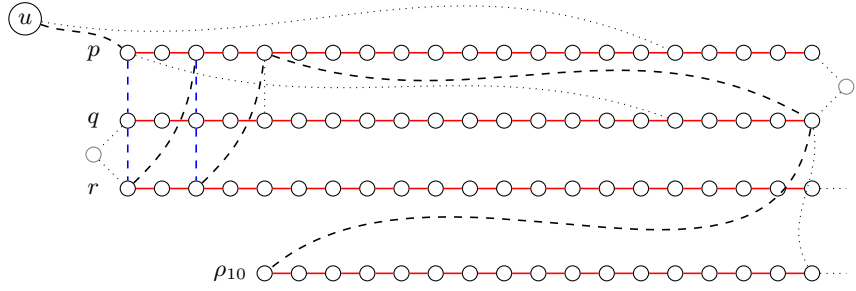
If the final circle is red then it is induced because the initial path is induced
and we neither create edges connecting two vertices of $\P_k$ to itself,
nor edges connecting $v_i$ to any vertices between endpoints of the cycle.
If the blue cycle is created it is induced because we use only odd vertices on $\spread[\P,9]$ for creating the three blue paths and no edges are created between vertices which are further than $1$ apart on these blue paths. 

Note that the length of paths $\spread[\P,9]$ is sufficient
because they need to be at least $2\big\lfloor\frac{n/2 - 1}{3}\big\rfloor+(n-2) \leq \frac{4n - 8}{3} \leq 4n/3-1$.

By Theorem~\ref{thm:indPaths} we can create the initial induced $P_{13n}$ in $28(13n)-27$ rounds.
There are at most $3n$ additional edges, hence $\ro(C_n) \leq \asEvenInd$.
\qed\end{proof}


\begin{theorem}\label{thm:oddcycles}
    Let $C_n$ be a cycle on $n$ vertices, where $n$ is odd.
    Then $\rsio(C_n) \le \rsio(C_{2n})+n \le \asOddInd$.
\end{theorem}

\begin{proof}
First, we create a monochromatic cycle $C_{2n}$.
Assume \WLOG that this cycle is blue.
Let $c_0,c_1,\dots,c_{2n-1}$ denote vertices on the $C_{2n}$ in the natural order
and let $c_i$ for any $i \ge 2n$ denote vertex $c_j$, $j = i \bmod 2n$.
We join two vertices which lie $n-1$ apart on the even cycle by creating an edge $\{c_0, c_{n-1}\}$.
If the edge is blue it forms a blue $C_n$ with part of the blue even cycle (see Fig.~\ref{fig:odd_cycles}).
If the edge is red we can continue and create an edge $\{c_{n-1},c_{2(n-1)}\}$ and use the same argument.
This procedure can be repeated $n$ times finishing with the edge $\{c_{(n-1)(n-1)},c_{n(n-1)}\}$
where $c_{n(n-1)}=c_0$ because $n-1$ is even.

Let $E$ be all the new red edges we just created, i.e.,
$E=\bigl\{\{c_{i},c_{i+n-1}\} \mid i \in J\bigr\}$ where $J=\bigl\{j(n-1) \mid j \in \{0,1,\dots,n-1\}\bigr\}$.
Since $\gcd(n-1,2n)=2$ it follows that the edges of $E$
complete a cycle $C^{'}_n=\bigl(\{c_0,c_2,\dots,c_{2n-2}\}, E\bigr)$ (see Fig.~\ref{fig:odd_cycles}).

Since the $C_{2n}$ is induced then it follows trivially that the target $C_n$ will be induced as well.

We used Theorem~\ref{thm:evencycles} to create an even cycle $C_{2n}$.
Then we added $n$ edges to form the $C^{'}$.
This gives us an upper bound for induced odd cycles $\rsio(C_n) \le \rsio(C_{2n})+n \le \asOddInd$.
\qed\end{proof}

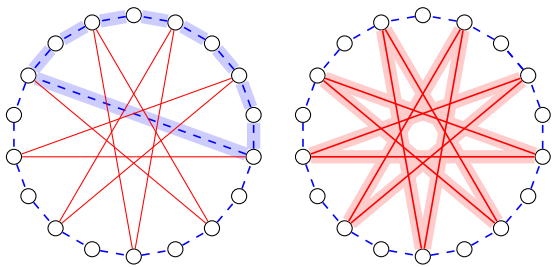
\begin{figure}[h]
    \centering
    \begin{tikzpicture}[scale=1.6]
        \begin{scope}[local bounding box=scope1]
        \node at (0,0){};
        \end{scope}
        \foreach \s[count=\si from 0] in {0,20,...,340}{
            \begin{scope}[shift={($(scope1) +(\s:0)$)}, scale=1,rotate=\s+30]
                \node[draw,circle,inner sep=2pt] (\si) at (1,0){};
            \end{scope}
        }
        \foreach \si in {16,...,23}{
            \draw[selected blue] let \n1={int(mod(\si+1,18))}, \n2={int(mod(\si,18))} in (\n2) -- (\n1);
            \draw[blue] let \n1={int(mod(\si+1,18))}, \n2={int(mod(\si,18))} in (\n2) -- (\n1);
        }
        \draw[selected blue] (16) -- (6);
        \draw[blue] (16) -- (6);
        \foreach \si in {24,...,33}{
            \draw[blue] let \n1={int(mod(\si+1,18))}, \n2={int(mod(\si,18))} in (\n2) -- (\n1);
        }
        \foreach \s[count=\si from 0] in {20,40,...,160}{
            \draw[draw=red] let \n1={int(mod(2*\si,18))},\n2={int(mod(2*\si+8,18))} in (\n1) -- (\n2);
        }
        \end{tikzpicture}
        ~~
        \begin{tikzpicture}[scale=1.6]
        \begin{scope}[local bounding box=scope1]
        \node at (0,0){};
        \end{scope}
        \foreach \s[count=\si from 0] in {0,20,...,340}{
            \begin{scope}[shift={($(scope1) +(\s:0)$)}, scale=1,rotate=\s+30]
                \node[draw,circle,inner sep=2pt] (\si) at (1,0){};
            \end{scope}
        }
        \foreach \s[count=\si from 0] in {0,20,...,360}{
            \draw[blue] let \n1={int(mod(\si+1,18))},\n2={int(mod(\si,18))} in (\n2) -- (\n1);
        }
        \foreach \s[count=\si from 0] in {0,20,...,160}{
            \draw[selected red] let \n1={int(mod(2*\si,18))},\n2={int(mod(2*\si+8,18))} in (\n1) -- (\n2);
        }
        \foreach \s[count=\si from 0] in {0,20,...,160}{
            \draw[red] let \n1={int(mod(2*\si,18))},\n2={int(mod(2*\si+8,18))} in (\n1) -- (\n2);
        }
    \end{tikzpicture}
    \caption{Final step of building $C_9$}
    \label{fig:odd_cycles}
\end{figure}

\subsection*{Non-induced Cycles}
Although the induced cycle strategies are asymptotically tight we can get better constants for the non-induced cycles.
For even cycles, we can use the non-induced path strategy to create the initial $P_{17n/2}$ in $4(17n/2)-3$ rounds.
Then we add $3n/2$ edges in the similar fashion as for the induced cycles
however we can squeeze them more tightly as depicted in the Fig.~\ref{fig:evencyclesaving}.
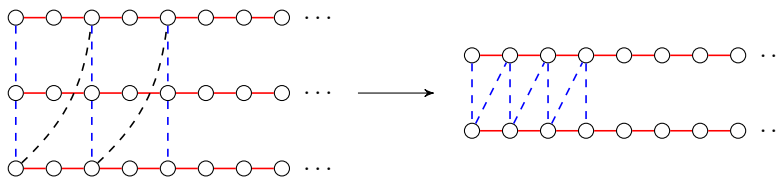
\begin{figure}[h]
    \centering
    \begin{tikzpicture}[scale=0.5]
        \newcommand{\drawtwopaths}[2]{
            \def\c{#1}
            \def\n{#2}
            \foreach \s[count=\si from 0] in {0,1,...,7}{
                \node[main] (\n\si\c) at (\s,0){};
            }
            \node[hide] at (8,0) {$\dots$};
            \foreach \s[count=\si from 0, count=\ss from 1] in {0,1,...,6}{
                \draw[red] (\n\si\c) -- (\n\ss\c);
            }
        }
        \begin{scope}[shift={(0,2)}] \drawtwopaths{1}{u} \end{scope}
        \begin{scope}[shift={(0,0)}] \drawtwopaths{1}{m} \end{scope}
        \begin{scope}[shift={(0,-2)}] \drawtwopaths{1}{d} \end{scope}
        \begin{scope}[shift={(12,1)}] \drawtwopaths{2}{u} \end{scope}
        \begin{scope}[shift={(12,-1)}] \drawtwopaths{2}{d} \end{scope}

        \draw[draw=black,->] (9,0) -- (11,0);
        \draw[blue] (u01) -- (m01) -- (d01) edge[bend right=20] (u21);
        \draw[blue] (u21) -- (m21) -- (d21) edge[bend right=20] (u41);
        \draw[blue] (u41) -- (m41) -- (d41);

        \draw[blue] (u02) -- (d02)
                 -- (u12) -- (d12)
                 -- (u22) -- (d22)
                 -- (u32) -- (d32);
    \end{tikzpicture}
    \caption{More efficient construction for even non-induced cycles.}
    \label{fig:evencyclesaving}
\end{figure}

Using this method the paths $\spread[\P,6]$ need only $5n/4$ vertices each, therefore the initial path $P_{17n/2}$ is sufficient.
This gives us $\ro(C_n) \leq 71n/2-3$ for even $n$ which directly translates to odd cycles
and gives us $\ro(C_n) \le \ro(C_{2n})+n \le 72n-3$ for odd $n$.

\egroup 

\bgroup 
\def\k{k}
\def\l{\ell}
\def\T{\mathcal{T}}
\def\C{C}
\def\S{S_{\k,\l}}

\section{Tight bounds for a family of trees}
\label{sec:spiders}

\bgroup 
\newcommand{\V}[2]{{P_{{#1},{#2}}}}
\def\Spider{{\sigma}}
\def\P{{\Spider_{\k,\l}}}
\def\PP{\mathbb{P}}
\def\resind{57\k^2\l+28\k^2-\k\l-27}
\def\resspd{\k^2\l+15\k\l+2\k-12}

We first prove a general lower bound for the online Ramsey number of graphs.
It will be used to show the tightness of bounds in this section.

\begin{lemma}\label{lem:lowerboundorn}
The $\ro(H)$ is at least $\mbox{VC}(H) \bigl(\Delta(H)-1\bigr)/2 + |E(H)|$ where $\mbox{VC}(H)$ is the vertex cover
and $\Delta(H)$ is the highest vertex degree in $H$ and $|E(H)|$ is the number of edges.
\end{lemma}
\begin{proof}
Let $deg_b(v)$ be the number of blue edges incident to the vertex $v$.
Let us define the Painter's strategy against the target graph $H$ as:
\begin{enumerate}
    \item if both incident vertices have $deg_b < \Delta(H)-1$ then color the edge blue,
    \item otherwise color the edge red.
\end{enumerate}
It is clear that Builder cannot create $H$ in blue color because the blue graph can contain only vertices with degree at most $\Delta(H)-1$.
To obtain a red edge it has to have at least one incident vertex with high blue degree.
The minimal number of vertices with high blue degree which are required to complete $H$
is the vertex cover of $H$, therefore, Builder has to create at least $\mbox{VC}(H)(\Delta(H)-1)/2$ blue edges.
Then Builder has to create at least $|E(H)|$ edges to complete the target graph in red color.
\qed\end{proof}

Let us define a \emph{spider} $\P$ for $\k \ge 3$ and $\l \ge 2$ as a union of $\k$ paths of length~$\l$
that share exactly one common endpoint.
Let a \emph{center} of $\P$ denote the only vertex with degree equal to $\k$.

In the following theorem we obtain an upper bound on $\ro(\P)$
that asymptotically matches the lower bound from Lemma~\ref{lem:lowerboundorn}.
\begin{theorem}\label{thm:induced spiders}
$\rsio(\P) = \Theta(\k^2 \l)$.
\end{theorem}

\begin{proof}
We describe Builder's strategy for obtaining an induced monochromatic $\P$.
We start by creating an induced monochromatic path of length $\k^2 (2\l+1)$ which is \WLOG blue.
This path contains $\k^2$ copies of $P_{2\l}$ as an induced subgraph.
Let $\V{i}{j}$ denote the $j$-th vertex on path $P_i$.
Let $\spreadu[\PP,\k]$ be $\k$ sets where each contains $\k$ disjoint induced paths.
Let $u$ be a previously unused vertex.
Now for each $\PP^j$ we do the following procedure:

\begin{enumerate}
    \item Let $\{\spreadu[P,\k]\}=\PP^j$.
    \item
        Create edges $\bigl\{ \{ u, w \} \mid w \in \{\spreadu[P_1,\k]\} \bigr\}$.
        If there are $\k$ blue edges there is a $\P$ with the center in $u$.
        If that is not the case there is at least one red edge $e^1=\{u,v^1\}$ where $v^{1} \in \{\spreadu[P_1, \k]\}$.
    \item For $i$ from $2$ to $\l$ we do as follows.
        \begin{itemize}
            \item
                For $v^{i-1} \in P^z$ create edges $\bigl\{\{v^{i-1},w\} \mid w \in \{\spreadu[P_i,\k]\} - P_i^z\bigr\}$.
                If all of these edges are blue we have a $\P$ with the center in $v^{i-1}$,
                otherwise there is a red edge $\{v^{i-1}, v^{i}\}$ where $v^{i} \in \{\spreadu[P_i, \k]\}$.
        \end{itemize}
    \item We obtained a red induced path $L^j = \bigl(u,\{u,v^1\},\dots,v^\l\bigr)$.
\end{enumerate}

If all iterations end up in obtaining a path $L^j$ we have $\k$ induced paths of length~$\l$
which all start in $u$ and together they form a $\P$ with the center in $u$.

We built a path $P_{\k^2(2\l+1)}$ using Theorem~\ref{thm:indPaths} using at most $28\bigl(\k^2(2\l+1)\bigr)-27$ edges.
During iterations, we created at most $\k\l(\k-1)$ edges.
Therefore we either got a blue $\P$ during the process or a red $\P$ after using no more than $\rsio(\P) \leq \resind = O(\k^2\l)$ rounds.

The lower bound of Lemma~\ref{lem:lowerboundorn} gives us $\Omega(\k^2\l)$ therefore the $\rsio(\P)=\Theta(\k^2\l)$.
\qed\end{proof}

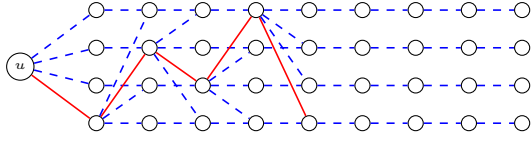
\begin{figure}[h]
    \centering
    \begin{tikzpicture}[scale=0.5]
    \tikzstyle{every node}=[font=\tiny, main]
    \begin{scope}[local bounding box=scope0, shift={(-6,0)}]
        \def\h{3}
        \def\w{8}
        \node[main] (s) at (-2, 0.5*\h) {$u$};
        \foreach \y in {0,...,\h}{
            \foreach \x in {0,...,\w}{
                \node[main] (\x\y) at (1.4*\x,\y){};
            }
            \foreach \t[count=\f from 0] in {1,...,\w}{
                \draw[blue] (\f\y) -- (\t\y);
            }
        }
        \def\wp{4}
        \draw[red ] (s)  -- (00);
        \draw[blue] (s)  -- (01);
        \draw[blue] (s)  -- (02);
        \draw[blue] (s)  -- (03);

        \draw[blue] (00) -- (11);
        \draw[red ] (00) -- (12);
        \draw[blue] (00) -- (13);

        \draw[blue] (12) -- (20);
        \draw[red ] (12) -- (21);
        \draw[blue] (12) -- (23);

        \draw[blue] (21) -- (30);
        \draw[blue] (21) -- (32);
        \draw[red ] (21) -- (33);

        \draw[red ] (33) -- (40);
        \draw[blue] (33) -- (41);
        \draw[blue] (33) -- (42);

    \end{scope}
\end{tikzpicture}
    \caption{Building one red leg of a spider $\Spider_{4,5}$.}
    \label{fig:spider}
\end{figure}

We can get the bound on non-induced spiders in a similar way,
however, we can use several tricks to get a bound which is not far from the lower bound.

\begin{theorem}\label{thm:spiders}
$\ro(\P) \le \resspd = O(\k^2 \l)$.
\end{theorem}

\begin{proof}
We create a path $P_{4\k\l}$ using strategy by Grytczuk et al.~\cite{Grytczuk2008OnlineRN}
in $4(4\k\l)-3$ rounds and split it into $2\k$ paths of length~$2\l$.
We follow the same strategy as in the induced case,
however, we work over the same set of paths in all iterations and we exclude those vertices which are already used by some path.
Choosing $2\k$ paths guarantees that we have big enough set even for the last iteration.
We create $2\k$ edges from $u$ and then we use $\k\l(\k-1)$ to create the red paths.
We either get a blue $\P$ in the process or a red $\P$ after using no more than $\resspd$ rounds.
\qed\end{proof}

\egroup 

\section{Family of induced trees with an asymptotic gap}
\label{sec:trees}
\newcommand{\odeg}{\overline{\deg}}
\newcommand{\beforerestrees}{426\k\l-444\k+280\l-296}
\newcommand{\restrees}      {426\k\l-442\k+308\l-295} 
In 2009 Conlon~\cite{Conlon2009OnlineRN} showed that the online Ramsey number and the size-Ramsey number
differ asymptotically for an infinite number of cliques.
In this section, we present a family of trees which exhibit the same property, i.e.,
their induced online Ramsey number and size-Ramsey number differ asymptotically.

\begin{definition}
Let the centipede $\S$ be a tree consisting of a path $P_\l$ of length~$\l$ where each of its vertices is center of star $S_k$,
i.e., a thorn-regular caterpillar.
\end{definition}

Note that $\S$ has $(\k+1)(\l+1)$ vertices and its maximum degree is $\k+2$.
We will show that $\S$ exhibits small induced online Ramsey number.

%

\newcommand{\QR}{Q_r}
\newcommand{\QB}{Q_b}

First, we need some ``degree-type'' notion.
Let $G = (V,E)$ be a graph whose edges are colored red and blue.
Let $U \subseteq V$.
For a vertex $v \in V$ let $\odeg(v,U)$ be a degree outside $U$.
Formally, $\odeg(v,U) = |N(v) \setminus U|$, where $N(v)$ is a neighborhood of $v$.
Let $\odeg_b(v,U)$ and $\odeg_r(v, U)$ be a vertex degree outside $U$ in blue or red color, respectively.
I.e.,
\[
    \odeg_b(v,U) = \Bigl|\bigl\{ u \in N(v) \setminus U : \{u,v\} \text{ is a blue edge} \bigr\}\Bigr|.
\]
and similarly for $\odeg_r(v,U)$.

A \emph{center} of a star $S_k$ is the vertex of degree $k$.
A \emph{center} of union of stars are centers of all stars in the union.
A \emph{colorful star} is a star such that for its center $v$ holds
that $\odeg_b(v) \geq k$ and $\odeg_r(v) \geq k$.
Let $H$ be a centipede or a union of stars. We denote a center of $H$ by $c(H)$.

\begin{theorem}\label{thm:trees}
$\rsio(\S) \le \restrees = O(\k\l)$.
\end{theorem}
%

\begin{proof}

We will proceed in steps where each step will get us closer to getting the result.
Let a superscript $X^i$ of any set $X$ denote the state of the set in $i$-th step.
Also, let $X^{i+1}=X^i$ if not mentioned otherwise.

We will gradually build two centipedes (one red, one blue) and a set of colorful stars.
Let $R^i$ ($B^i$) be a red (blue) centipede in the step $i$. 
First, we assume that both $R^i$ and $B^i$ are nonempty. 
We show later a strategy for the case $R^i$ or $B^i$ is empty (i.e., centipede of length 0).

Let $\QR^i$ be a union of colorful stars such that for each star $S \in \QR^i$ holds that $c(S) \in c(R^j)$ for some $j < i$, i.e. the center of $S$ were in the center of the red centipede in some previous step.
The $\QB^i$ is defined similarly, i.e. it is a union of colorful stars such that for each star $S \in \QB^i$ holds that $c(S) \in c(B^j)$ for some $j < i$.
Let $U^i = R^i \cup B^i \cup \QR^i \cup \QB^i$.
For $v \in c(R^i)$ let $\odeg_o(v)$ be $\odeg_b(v, U^i)$, i.e. blue degree of $v$ outside centipedes and colorful stars.
Similarly, let $\odeg_o(v)$ be $\odeg_r(v, U^i)$ for $v \in c(B^i)$.
Each step we either make one centipede longer by $1$, add one colorful star to $\QR$ or $\QB$ or increase $\odeg_o(v)$ of $v \in c(R^i) \cup c(B^i)$.
One step will proceed as follows:

\begin{enumerate}
    \item Let $u$ and $v$ be endpoints of $c(R^i)$ and $c(B^i)$ respectively.
    \item Create an edge $e=\{u,x\}$ where $x$ is previously unused vertex.
    \item If $e$ is blue set $w:=u$, if $e$ is red create an edge $f=\{v,x\}$ and set $w:=v$.
    \item \label{case:three} Perform one of the following:
        \begin{enumerate}
            \item \label{case:forkok} If $e$ is red and $f$ is blue,
	    create edges from $x$ until $k$ of them are in the same color and then add $x$ to respective centipede center set.
            \item \label{case:forknok} Either $e$ is blue, or both $e$ and $f$ are red,
                \begin{enumerate}
                    \item \label{case:forknoksmall} if $\odeg_o(w)<k$, the $\odeg_o(w)$ was increased by $1$,
                    \item \label{case:forknokbig} or $\odeg_o(w) \geq k$, we have a colorful star with center in $w$,
		    therefore we move $w$ from its centipede center set to respective colorful star set, i.e.,
                    $c(\QR^{i+1})=c(\QR^i) \cup \{u\}$ and $c(R^{i+1})=c(R^i) - u$ if $w=u$,
		    or $c(\QB^{i+1})=c(\QB^i) \cup \{v\}$ and $c(B^{i+1})=c(B^i)-v$ if $w=v$.
                \end{enumerate}
        \end{enumerate}
\end{enumerate}
See Figure~\ref{fig:star_star} for clarification of various cases during one step.
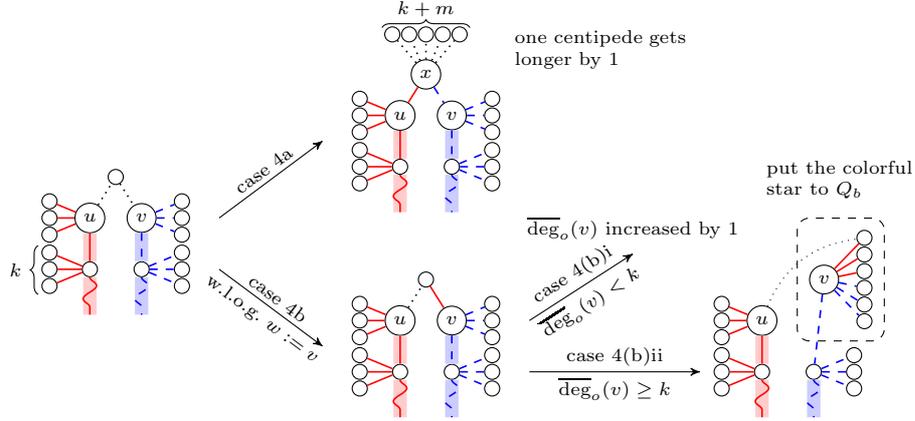
\begin{figure}[h]
    \centering
    \begin{tikzpicture}[
            every node/.style = {main},
            level 1/.style = {sibling distance=1em},
            level distance = 2.4em,
            scale=0.68,
            font=\scriptsize,
        ]
        \newcommand{\smallstar}[5]{
            \begin{scope}[grow=#2, shift={#4}]
                \node (#3) {#5}
                    child {node {} edge from parent [#1] }
                    child {node {} edge from parent [#1] }
                    child {node {} edge from parent [#1] }
                ;
            \end{scope}
        }
        \newcommand{\partcatepillar}{
            \smallstar{red}{left}{u}{(0,2)}{$u$}
            \smallstar{red}{left}{u1}{(0,1)}{}
            \node[hide] (u2) at (0,0){};
            \draw[selected red] (u) -- (u1);
            \draw[red] (u) -- (u1);
            \draw[selected red] (u1) -- (u2);
            \draw[red,longpath] (u1) -- (u2);
            \smallstar{blue}{right}{v}{(1,2)}{$v$}
            \smallstar{blue}{right}{v1}{(1,1)}{}
            \node[hide] (v2) at (1,0){};
            \draw[selected blue] (v) -- (v1);
            \draw[blue] (v) -- (v1);
            \draw[selected blue] (v1) -- (v2);
            \draw[blue,longpath] (v1) -- (v2);
        }
        \begin{scope}[local bounding box=a, shift={(0,0)}]
            \draw[decorate,decoration={brace,amplitude=3pt},xshift=0pt,yshift=6pt] (-1,0.3) -- (-1,1.25) node [hide,midway,xshift=-0.3cm] {$k$};
            \partcatepillar
            \node (x) at (0.5,2.8){};
            \draw[tobe] (u) -- (x) -- (v);
        \end{scope}
        \begin{scope}[local bounding box=b, shift={(6,2)}]
            \partcatepillar
            \begin{scope}[grow=up, shift={(0.5,2.8)}]
                \node (x) {$x$}
                    child {node {} edge from parent [tobe] }
                    child {node {} edge from parent [tobe] }
                    child {node {} edge from parent [tobe] }
                    child {node {} edge from parent [tobe] }
                    child {node {} edge from parent [tobe] }
                ;
            \end{scope}
            \draw[red]  (u) -- (x);
            \draw[blue] (x) -- (v);
            \draw[decorate,decoration={brace,amplitude=3pt},xshift=0pt,yshift=6pt] (-0.35,3.5) -- (1.35,3.5) node [hide,midway,yshift=0.25cm] {$k+m$};
            \node[hide,text width=7em,inner sep=0.5em] () [right of=b,anchor=west,shift={(0,0.8)}] {one centipede gets longer by $1$};
        \end{scope}
        \begin{scope}[local bounding box=c, shift={(6,-2)}]
            \partcatepillar
            \node (x) at (0.5,2.8){};
            \draw[tobe] (u) -- (x);
            \draw[red] (x) -- (v);
        \end{scope}
        \begin{scope}[local bounding box=d, shift={(13,-2)}]
            \smallstar{red}{left}{u}{(0,2)}{$u$}
            \smallstar{red}{left}{u1}{(0,1)}{}
            \node[hide] (u2) at (0,0){};
            \draw[selected red] (u) -- (u1);
            \draw[red] (u) -- (u1);
            \draw[selected red] (u1) -- (u2);
            \draw[red,longpath] (u1) -- (u2);
            \smallstar{blue}{right}{v1}{(1,1)}{}
            \begin{scope}[local bounding box=dd, grow=right, shift={(1.2,2.8)}]
                \node (v) {$v$}
                    child {node {} edge from parent [blue] }
                    child {node {} edge from parent [blue] }
                    child {node {} edge from parent [blue] }
                    child {node {} edge from parent [red] }
                    child {node {} edge from parent [red] }
                    child {node (x) {} edge from parent [red] }
                ;
            \end{scope}
            \node[hide] (v2) at (1,0){};
            \draw[blue] (v) -- (v1);
            \draw[selected blue] (v1) -- (v2);
            \draw[blue,longpath] (v1) -- (v2);
            \draw[dotted] (u) edge [bend left] (x);
            \node[draw=black,rectangle, dashed,fit=(dd),rounded corners=0.5em,inner sep=0.5em](box) {};
            \node[hide,text width=6em] () [above of=box,anchor=south]{put the colorful star to $\QB$};
        \end{scope}
        \begin{scope}[local bounding box=lc, shift={(10.5,1.8)}]
            \node[hide] () {$\odeg_o(v)$ increased by $1$};
        \end{scope}
        \draw[->] (2.5,2) -- node[above,hide,sloped]{case~\ref{case:forkok}} (4.5,3.5);
        \draw[->] (2.5,1) -- node[above,hide,sloped]{case~\ref{case:forknok}}
                             node[below,hide,sloped]{w.l.o.g. $w:=v$} (4.5,-0.5);
        \draw[->] (8.5,0) -- node[above,hide,sloped]{case~\ref{case:forknoksmall}}
                             node[below,hide,sloped]{$\odeg_o(v) < k$} (10.5,1.3);
        \draw[->] (8.5,-1) -- node[above,hide]{case~\ref{case:forknokbig}}
                                node[below,hide]{$\odeg_o(v) \geq k$} (11.8,-1);
    \end{tikzpicture}
    \caption{One step of building a $\S$ where $k=3$}
    \label{fig:star_star}
\end{figure}

Let $p^i$ be a potential in step $i$ defined as
\[
    p^i=\bigl(|c(R^i)|+|c(B^i)|\bigr)\bigl(k+2\bigr)
       + \bigl(|c(\QR^i)|+|c(\QB^i)|\bigr)\bigl(3k+2\bigr)
       + 2\smashoperator{\sum_{v \in c(R^i) \cup c(B^i)}}{\odeg_o(v)}.
\]
Note that for all the outcomes of one step the potential will increase by at least the number of created edges.
\begin{itemize}
    \item In case~\ref{case:forkok} we create $2+k+m$ edges. $k+1$ edges extend one centipede by one star, one edge is not used, and $m \leq k-1$ edges are additional edges of the other color on that star. Extending one centipede by a star with $m$ edges in other color increases $p$ by $(k+2)+2m$.
    \item In case~\ref{case:forknoksmall} we create at most $2$ edges, increasing $\odeg_o$ of one vertex by one, which increases $p$ by $2$.
    \item In case~\ref{case:forknokbig} we create at most $2$ edges, making one centipede shorter by one, however adding one colorful star to either $\QR$ or $\QB$ so $p$ increases by $(3k+2)-(k+2)-2(k-1)=2$.
\end{itemize}
Note that the graphs induced by $c(R^i)$ and $c(B^i)$ respectively are paths.
These graphs are altered by adding one vertex at the end or moving end-vertex to respective $c(Q^i)$ set.
It follows that the graphs induced by $c(\QR^i)$ and $c(\QB^i)$ are both forests.

Assume that after many steps we end up with $\bigl|c(R^i)\bigr|=\bigl|c(B^i)\bigr|=\l$, $\bigl|c(\QR^i)\bigr| = \bigl|c(\QB^i)\bigr| = 2(35\l-36)-2$, $\odeg_o(v)=k-1$ for all $v \in c(R^i)\cup c(B^i)$, and we did not win yet.
In such situation the potential is
\begin{align*}
p^i &= 2\l(k+2)+2\bigl(2(35\l-36)-2\bigr)(3k+2) + 2(k-1)2\l \\
    &= \beforerestrees.
\end{align*}
We now perform one last step in which we might win, but if not then either $\QR^{i+1}$ or $\QB^{i+1}$ will have $2(35\l-36)-1$ colorful stars.
We take the $35\l-36$ independent colorful star centers of the bigger $Q^i$ set and perform the induced path strategy on them, which guarantees a monochromatic centipede. 

The final step might add various number of rounds to our strategy depending on the case which we end up in.
Case~\ref{case:forkok} demands at most $1+2k$ edges and we win.
Case~\ref{case:forknoksmall} cannot happen because $\odeg_o(v)=\odeg_o(u)=k-1$.
And case~\ref{case:forknokbig} demands that we add $2$ edges and then we perform the path strategy using at most $\rsio(P_\l) \leq 28\l - 27$ edges.
We get the final upper bound on the number of edges $\rsio(\S) \leq \restrees$.

We now discuss why the final centipede is induced.
First, let us partition all vertices used in the strategy into three groups: $\mathcal{R}=c(R) \cup c(\QR)$, $\mathcal{B}=c(B) \cup c(\QB)$, and $\mathcal{O}$ (which contains all the remaining vertices).
Note that in each step some vertices are added to the groups but once assigned they never change their group.
Vertices in $\mathcal{R}$ and $\mathcal{B}$ are always added to $c(R)$ or $c(B)$ and then they might be moved into $c(\QR)$ and $c(\QB)$ respectively.
Vertices in $\mathcal{O}$ are used during one step and are never used again, specifically in case~\ref{case:forkok} there are $k+m$ vertices created and all of them are connected to $1$ center vertex (in $c(R)$ or $c(B)$) and in case~\ref{case:forknok} one new vertex is connected to at most one vertex from $\mathcal{R}$ and one vertex from $\mathcal{B}$.
Assume \WLOG that the centipede is in red color.
The centipede either appears with centers in $c(R)$ or $c(\QR)$.
Assume the former occurred then the centers of $c(R)$ induce a path.
If the latter occurred then the vertices of $c(\QR)$ we used in the induced path strategy were independent.
In both cases, the leaves of the centipede appear in the $\mathcal{O}$.
These vertices have at most one edge to the $\mathcal{R}$ and have no edges among each other.

If $R$ or $B$ is empty then the strategy changes slightly.
In all the cases we omit creation of edge $e$ if $R=\emptyset$ and $f$ if $B=\emptyset$.
If $e$ is omitted assume it is red when deciding what edges to draw, and respectively when $f$ is omitted assume it is blue.
We observe that all the steps stay the same and the potential increases in the same manner
but we created fewer edges than necessary which does not contradict the devised upper bound.
\qed\end{proof}

Due to Beck~\cite{Beck1990size} we have a lower bound for trees $T$
which is $\rf(T) \ge \beta(T)/4$ where $\beta(T)$ is defined as
\[
    \beta(T)=|T_0|\Delta(T_0)+|T_1|\Delta(T_1),
\]
where $T_0$ and $T_1$ are partitions of the unique bipartitioning of the tree $T$.
The $\beta$ for our family of trees is $\beta(\S) \approx \left({\l/2} + {\k\l/2} \right)(\k+2) = \Theta(\k^2\l)$,
which gives us the lower bound on size-Ramsey number $\rf(\S)=\Omega(\k^2\l)$.

Since by Theorem~\ref{thm:trees} we have $\ro(\S)\le\rsio(\S)=O(\k\l)$ the online Ramsey number for $\S$
is asymptotically smaller than its size-Ramsey number.

\begin{corollary}\label{cor:trees}
There is an infinite sequence of trees $T_1, T_2, \dots$ such that $|T_i| < |T_{i+1}|$ for each $i \ge 1$ and
\[
    \lim_{i\to\infty} \frac{\ro(T_i)}{\rf(T_i)} = 0.
\]
\end{corollary}

\egroup 

\section*{Acknowledgments}
We would like to thank our colleagues Jiří Fiala, Pavel Veselý and Jana Syrovátková for fruitful discussions.

\bibliographystyle{plain}
\bibliography{main}

\end{document}